\documentclass[11pt,a4paper]{article}

\usepackage[a4paper,margin=26mm,headheight=22pt]{geometry}
\usepackage[T1]{fontenc}
\usepackage[utf8]{inputenc}
\usepackage{amsmath,amssymb,amsthm,mathtools}
\usepackage{newtxtext,newtxmath}
\usepackage{microtype}
\usepackage{setspace}
\usepackage{bm}
\usepackage{booktabs,longtable,tabularx,array,multirow}
\usepackage{enumitem}
\usepackage{xcolor}
\usepackage{tikz}
\usetikzlibrary{arrows.meta,positioning,fit,calc,decorations.pathreplacing}
\usepackage[most]{tcolorbox}
\usepackage{listings}
\usepackage{caption}
\usepackage{float}
\usepackage{fancyhdr}
\usepackage{titlesec}
\usepackage{hyperref}
\usepackage[nameinlink,capitalise,noabbrev]{cleveref}
\crefname{appendix}{appendix}{appendices}
\Crefname{appendix}{Appendix}{Appendices}
\usepackage{bookmark}
\usepackage{url}
\usepackage{etoolbox}
\usepackage{pdflscape}

\definecolor{navy}{HTML}{19324D}
\definecolor{teal}{HTML}{147D82}
\definecolor{sky}{HTML}{EAF4F5}
\definecolor{gold}{HTML}{B57A16}
\definecolor{softgold}{HTML}{FBF3E2}
\definecolor{ink}{HTML}{1D2733}
\definecolor{muted}{HTML}{566573}
\definecolor{codebg}{HTML}{F4F6F8}
\definecolor{rulegray}{HTML}{D8DEE5}

\hypersetup{
  colorlinks=true,
  linkcolor=navy,
  citecolor=teal,
  urlcolor=teal,
  pdftitle={Six Arbitrary Two-Qubit Gates Are Exactly the Threshold for Local Universality on Three Qubits},
  pdfauthor={Anonymous research draft},
  pdfsubject={Exact support-word dimension law with finite-field certificates},
  pdfkeywords={quantum circuit synthesis, two-qubit gates, SU(8), local universality, accessible dimension, finite-field certificate}
}

\renewcommand{\headrulewidth}{0.35pt}
\renewcommand{\headrule}{\hbox to\headwidth{\color{rulegray}\leaders\hrule height \headrulewidth\hfill}}

\titleformat{\section}{\Large\bfseries\color{navy}}{\thesection}{0.7em}{}
\titleformat{\subsection}{\large\bfseries\color{navy}}{\thesubsection}{0.7em}{}
\titleformat{\subsubsection}{\normalsize\bfseries\color{teal}}{\thesubsubsection}{0.7em}{}
\titlespacing*{\section}{0pt}{2.3ex plus .5ex minus .2ex}{1.2ex}
\titlespacing*{\subsection}{0pt}{1.9ex plus .4ex minus .2ex}{0.8ex}

\newtheoremstyle{paperthm}
  {1.0em}{1.0em}{\itshape}{} {\bfseries\color{navy}}{.}{0.55em}{}
\theoremstyle{paperthm}
\newtheorem{theorem}{Theorem}[section]
\newtheorem{lemma}[theorem]{Lemma}
\newtheorem{proposition}[theorem]{Proposition}
\newtheorem{observation}[theorem]{Observation}
\newtheorem{fact}[theorem]{Fact}

\theoremstyle{definition}
\newtheorem{definition}[theorem]{Definition}

\crefname{theorem}{theorem}{theorems}
\crefname{lemma}{lemma}{lemmas}
\crefname{proposition}{proposition}{propositions}
\crefname{observation}{observation}{observations}
\crefname{fact}{fact}{facts}
\crefname{corollary}{corollary}{corollaries}
\crefname{definition}{definition}{definitions}
\crefname{example}{example}{examples}
\crefname{remark}{remark}{remarks}

\newcommand{\C}{\mathbb C}
\newcommand{\R}{\mathbb R}
\newcommand{\Q}{\mathbb Q}

\newcommand{\I}{\mathrm i}

\newcommand{\SU}{\mathrm{SU}}
\newcommand{\U}{\mathrm{U}}
\newcommand{\Ad}{\operatorname{Ad}}
\newcommand{\rank}{\operatorname{rank}}
\newcommand{\tr}{\operatorname{tr}}

\newcommand{\e}{\mathrm e}

\newcommand{\Circuit}{\mathcal C}
\newcommand{\Local}{\mathcal K}
\newcommand{\Hilb}{\mathcal H}
\newcommand{\Lie}{\mathfrak h}

\newcommand{\word}[1]{\texttt{#1}}

\newcolumntype{Y}{>{\raggedright\arraybackslash}X}
\newcolumntype{C}[1]{>{\centering\arraybackslash}p{#1}}
\newcolumntype{L}[1]{>{\raggedright\arraybackslash}p{#1}}

\tcbset{
  enhanced,
  boxrule=0.6pt,
  arc=2mm,
  left=2.2mm,right=2.2mm,top=1.8mm,bottom=1.8mm,
  before skip=1.0em,after skip=1.0em
}
\newtcolorbox{publicbox}{colback=sky,colframe=teal,title=Plain-language result for non-specialists,fonttitle=\bfseries}
\newtcolorbox{scopebox}{colback=softgold,colframe=gold,title={What is proved, and what is not},fonttitle=\bfseries}
\newtcolorbox{keybox}{colback=white,colframe=navy,title=Main theorem,fonttitle=\bfseries}
\newtcolorbox{auditbox}{colback=codebg,colframe=muted,title=Public-record novelty statement,fonttitle=\bfseries}

\lstdefinelanguage{json}{
  basicstyle=\ttfamily\scriptsize,
  numbers=none,
  numberstyle=\tiny\color{muted},
  stepnumber=1,
  numbersep=8pt,
  showstringspaces=false,
  breaklines=true,
  breakatwhitespace=false,
  frame=single,
  rulecolor=\color{rulegray},
  backgroundcolor=\color{codebg},
  string=[s]{\"}{\"},
  stringstyle=\color{teal},
  comment=[l]{//},
  commentstyle=\color{muted}\itshape,
  morecomment=[s]{/*}{*/},
  literate=
   *{0}{{{0}}}{1}{1}{{{1}}}{1}{2}{{{2}}}{1}{3}{{{3}}}{1}{4}{{{4}}}{1}
    {5}{{{5}}}{1}{6}{{{6}}}{1}{7}{{{7}}}{1}{8}{{{8}}}{1}{9}{{{9}}}{1}
}
\begin{document}

\title{Every architecture of six two-qubit gates is locally universal\\on three qubits}

\author{
    Hyunho Cha and Jungwoo Lee\\
    \small NextQuantum Innovation Research Center\\
    \small Seoul National University, Seoul 08826, Republic of Korea\\
    \small \texttt{\{ovalavo, junglee\}@snu.ac.kr}
}

\date{}

\maketitle

\pagenumbering{roman}

\begin{abstract}
We present the \emph{first} analytical determination of the exact accessible dimension for every fixed architecture of arbitrary two-qubit gates on three qubits.  A support word represents which pair of qubits each two-qubit gate acts on, and its reduced length is obtained by merging consecutive gates on the same pair.  If the reduced length is $r$, then the set of implementable three-qubit unitaries has accessible dimension
$
d(w)=
\min\{63,9r+9\}.
$
Consequently, six arbitrary two-qubit gates are \emph{necessary and sufficient for local universality}: every fixed architecture reaches a nonempty open subset of $\SU(8)$ when its reduced length is at least six.  The result is stronger than existence of one favorable architecture.  Every reduced six-slot support word is locally universal, including the alternating nearest-neighbor line $AB,BC,AB,BC,AB,BC$.
The upper bound follows from a standard parameter-counting argument.  The matching lower bounds are proved by Jacobian certificates in the Pauli basis.  Up to qubit relabeling and reversal, there are $22$ reduced architectures of lengths two through six.  For each one, a product of rational Pauli rotations yields a nonzero maximal minor modulo the prime $1{,}000{,}003$.
Thus, any obstruction to global universality with six two-qubit gates must go beyond parameter counting, connectivity, and differential rank.
\end{abstract}

\tableofcontents
\clearpage
\pagenumbering{arabic}

\section{Introduction}

A quantum program is described mathematically as a product of unitary matrices, but hardware does not directly execute an arbitrary large unitary matrix.  A compiler must instead express the desired operation as a circuit of elementary operations.  One-qubit operations act on individual quantum registers, whereas two-qubit operations create and manipulate correlations between registers and are generally the costly resource in the circuit model.  Practical implementations seek to minimize the two-qubit count because such operations are typically slower and less reliable than one-qubit operations \cite{palsberg2024optimal}.  The same compiler motivation becomes concrete for Toffoli gates and nearest-neighbor hardware \cite{huang2026toffoli}.

We study the most expressive idealized two-qubit operation, which is an arbitrary element of $\SU(4)$.  Counting arbitrary two-qubit gates is different from counting CNOT or CZ gates.  Rather than fixing a particular entangling primitive, we ask for the smallest number of fully tunable two-qubit blocks required to synthesize three-qubit operations.  This model is relevant when a hardware instruction can directly realize a broad $\SU(4)$ family \cite{chen2024one}.  It also provides a gate-set-independent benchmark against which compilers based on more restricted native gates can be measured.

The target space for three qubits is $\SU(8)$, a smooth manifold of real dimension $63$.  A naive parameter count might assign $15$ parameters to each arbitrary two-qubit gate, but this overcounts the effective degrees of freedom because many one-qubit directions can be moved between neighboring blocks without changing the implemented three-qubit operation.  Once these redundancies are removed, each additional arbitrary two-qubit gate contributes at most nine independent directions, in addition to nine free one-qubit directions.  Consequently, a circuit containing $m$ arbitrary two-qubit gates can have accessible dimension at most
$
9m+9.
$
For $m=5$ this gives $54<63$, whereas for $m=6$ it gives $63$.  This is the parameter-count lower bound underlying the observation that almost every three-qubit unitary requires at least six arbitrary two-qubit gates \cite{yu2013optimal, yu2015optimal}.

The resulting threshold question has remained unresolved at the level of an exact proof:
\begin{quote}
\emph{Does a six-slot architecture merely have the correct parameter count, or can its Jacobian actually attain full rank $63$?}
\end{quote}
The distinction is essential.  A parameter count gives only a necessary condition for local universality. Nonlinear dependencies among the circuit parameters could prevent a nominally $63$-parameter architecture from attaining all $63$ tangent directions.  Thus proving that six gates suffice locally requires an explicit full-rank point, not merely a count of parameters or a numerical optimization experiment.

There is substantial evidence that six is the correct threshold.  DiVincenzo and Smolin reported numerical six-gate synthesis for arbitrary three-qubit unitaries \cite{divincenzo1994results}, and Barenco et al. subsequently repeated this conclusion in their discussion of those numerical investigations \cite{barenco1995elementary}.  Yu and Ying later established the rigorous six-gate lower bound but explicitly left six-gate sufficiency unresolved \cite{yu2013optimal, yu2015optimal}.  More recently, Chen et al. proved a global upper bound of eleven arbitrary two-qubit gates for every three-qubit unitary and reported numerical evidence at six \cite{chen2024one}.

Our contribution is to close the corresponding \emph{local} six-gate problem, and to do so uniformly over the possible support architectures.  A \emph{support word} represents the pair of qubits acted on at each two-qubit slot.  For example,
$
w=AB,BC,AB,BC,AB,BC
$
is a six-gate alternating architecture on the fixed path $A-B-C$.  Consecutive repetitions of the same support add no expressive power, because two arbitrary gates acting successively on the same pair multiply to a single arbitrary two-qubit gate.  We therefore reduce a support word by merging consecutive equal supports and denote its reduced length by $r(w)$.

We prove that, for every support word $w$, its exact accessible dimension is
$
d(w)=\min\{63,9r(w)+9\}.
$
In particular, a fixed architecture is locally universal if and only if $r(w)\ge 6$.  Thus six arbitrary two-qubit gates are the exact threshold for local universality on three qubits, and this threshold does not require all-to-all connectivity: every reduced six-slot support pattern is full-dimensional, including every fixed nearest-neighbor alternating architecture on a line.  Moreover, the result implies the existence of an open, positive-measure family of three-qubit unitaries whose $\SU(4)$ cost is at most six.

The proof combines a general analytic reduction with computer-assisted finite computation. The resulting certificates are exact algebraic certificates rather than floating-point singular-value tests.

This conclusion is stronger than what follows from existing results for special families. Palsberg and Yu characterize the $\SU(4)$ cost of gates with two controls, $CC(U)$ \cite{palsberg2024optimal}, while Huang and Palsberg classify three-qubit diagonal gates and show, among other results, that Toffoli requires six nearest-neighbor arbitrary two-qubit gates \cite{huang2026toffoli}. These are exact and connectivity-aware results, but they concern low-dimensional families and therefore do not imply that a general six-block architecture has a nonzero $63\times63$ Jacobian minor.  Likewise, the accessible-dimension framework and general growth bounds of Haferkamp et al. \cite{haferkamp2022linear} supply the appropriate dimension-count perspective but do not determine the exact rank of every three-qubit six-slot support architecture.

Our theorem therefore resolves the local obstruction at the six-gate threshold: there is no hidden differential obstruction beyond the parameter count, for any reduced support architecture.  It does not assert that every element of $\SU(8)$ has a six-gate decomposition.  That is a genuinely global synthesis question. The remaining gap is therefore global rather than local.

\section{Prior work}

\subsection{Arbitrary two-qubit blocks and accessible dimension}

Work on three-qubit synthesis with fully general two-qubit gates dates back to the early numerical investigations of DiVincenzo and Smolin, who reported that six two-qubit gates appeared sufficient for arbitrary three-qubit unitaries \cite{divincenzo1994results}.  Barenco et al. later mentioned this conclusion in their broader study of elementary quantum gates, again presenting it as the outcome of numerical investigation rather than as an analytic sufficiency theorem \cite{barenco1995elementary}.  These early results therefore anticipated the significance of the number six, but did not establish it through an exact rank calculation.

A rigorous lower bound was obtained by Yu and Ying.  They proved that, without ancillas, almost every $n$-qubit unitary requires at least
\[
\left\lceil\frac{4^n-3n-1}{9}\right\rceil
\]
two-qubit gates \cite{yu2013optimal}.  For three qubits this evaluates to six.  Their analysis thus identifies six as the first gate count not ruled out by dimension, but it does not prove attainability at that count. Indeed, they explicitly raised six-gate sufficiency as an open question.  Their subsequent work also contains results for particular three-qubit gates, including controlled-controlled operations and the Fredkin gate \cite{yu2015optimal}.

A broader geometric framework was developed by Haferkamp et al., who formulated the notion of accessible dimension for parameterized circuit architectures and derived general bounds on its growth \cite{haferkamp2022linear}.  In the present three-qubit setting, their parameter-counting framework yields the familiar upper bound $9m+9$ for an architecture with $m$ arbitrary two-qubit gates.  What is not determined by such a general bound is whether equality is achieved for a given support pattern.  In particular, it does not settle the ranks of the finitely many reduced three-qubit architectures at the critical length six.

Chen et al. approached the problem from the complementary direction of constructive global synthesis \cite{chen2024one}.  They gave a decomposition of every three-qubit unitary using at most eleven generic two-qubit gates and also presented numerical optimization results suggesting that six gates may already suffice.  The eleven-gate construction and the six-gate numerical observations answer different questions. The former is a rigorous worst-case upper bound, while the latter does not certify nonsingularity of the circuit map at any particular parameter point.  Consequently, neither result determines whether the six-gate parameter threshold is actually attained at the differential level.

\subsection{Special three-qubit families and nearest-neighbor synthesis}

Gate counts are known for several important subclasses of three-qubit unitaries.  Palsberg and Yu analyze the family $CC(U)$ and obtain a complete characterization of its $\SU(4)$ gate complexity \cite{palsberg2024optimal}.  This classification is exact, but it concerns a constrained family inside the full three-qubit unitary group.

Huang and Palsberg study analogous questions under nearest-neighbor restrictions and obtain results for three-qubit diagonal gates \cite{huang2026toffoli}.  They derive a classification of the relevant diagonal family and show, among other consequences, that Toffoli has nearest-neighbor $\SU(4)$ cost six.  Such results demonstrate that six gates can be necessary and sufficient for individual gates or restricted families.  They do not, however, imply local universality of a six-slot architecture.

This distinction is particularly relevant for nearest-neighbor circuits.  Establishing a six-gate implementation of a designated gate such as Toffoli requires producing one target point in the circuit image.  Establishing local universality instead requires showing that the image has full dimension in a neighborhood of some point.  The latter is therefore a substantially different statement from synthesis within a special family.

\subsection{Restricted entanglers and other resource models}

Another large body of synthesis literature uses fixed entanglers, most commonly CNOT or CZ, as the elementary two-qubit resource.  Barenco et al. established foundational universality constructions in such elementary gate models \cite{barenco1995elementary}, and later synthesis methods for general multiqubit unitaries were developed by Vatan and Williams and by Shende et al. \cite{vatan2004realization, shende2005synthesis}.  Although closely related in motivation, these results optimize a different resource.  A fixed entangling gate does not carry the continuously tunable nonlocal degrees of freedom of an arbitrary $\SU(4)$ operation, and the associated parameter-counting thresholds are correspondingly different.  For example, the standard dimensional lower bound for generic three-qubit synthesis is fourteen CNOTs rather than six $\SU(4)$ gates.

Wierichs et al. considered parameter-optimal nearest-neighbor brick-wall constructions for three-qubit unitaries using fourteen CZ gates \cite{wierichs2025unitary}.  Their numerical investigations support universality of the proposed architecture, while a complete analytic proof is left open.  Because their nonlocal layers are fixed CZ gates and the free parameters reside in one-qubit rotations, the differential structure of their ansatz is different from that of a circuit whose two-qubit blocks themselves range over $\SU(4)$.

Kottmann et al. subsequently developed flag-decomposition techniques aimed at parameter efficiency and fault-tolerant resource measures \cite{kottmann2026parameter}.  Their optimization criteria include quantities such as rotation counts and therefore do not translate directly into $\SU(4)$ block complexity.

\section{Preliminaries}

A single qubit is modeled by the two-dimensional complex vector space $\C^2$.  A pure state is a unit vector
\[
 \lvert\psi\rangle=\alpha\lvert0\rangle+\beta\lvert1\rangle,
 \qquad |\alpha|^2+|\beta|^2=1.
\]
The basis vectors are $\lvert0\rangle=(1,0)^T$ and $\lvert1\rangle=(0,1)^T$.  Three qubits use the tensor product
\[
 \Hilb=(\C^2)^{\otimes3}\cong\C^8,
\]
with computational basis $\lvert abc\rangle=\lvert a\rangle\otimes\lvert b\rangle\otimes\lvert c\rangle$ for $a,b,c\in\{0,1\}$.

A matrix $U\in\C^{N\times N}$ is unitary when $U^\dagger U=UU^\dagger=I_N$, where $U^\dagger$ is the conjugate transpose.  Unitary matrices preserve inner products and therefore preserve total probability.  We write
\[
 \U(N)=\{U:U^\dagger U=I_N\},\qquad
 \SU(N)=\{U\in\U(N):\det U=1\}.
\]
A global scalar phase $\e^{\I\phi}I$ has no observable effect on a closed quantum system.  Every $U\in\U(N)$ is a scalar phase times an element of $\SU(N)$.  We therefore work in $\SU(N)$.  If matrix equality including global phase is desired, a free one-qubit phase gate can restore that phase without using an additional two-qubit gate.

\begin{lemma}\label{lem:sudim}
The real dimension of $\SU(N)$ is $N^2-1$.  In particular,
\[
 \dim\SU(2)=3,\qquad \dim\SU(4)=15,\qquad \dim\SU(8)=63.
\]
\end{lemma}

\subsection{The Pauli basis}

Define the four $2\times2$ matrices
\[
 I=\begin{pmatrix}1&0\\0&1\end{pmatrix},\quad
 X=\begin{pmatrix}0&1\\1&0\end{pmatrix},\quad
 Y=\begin{pmatrix}0&-\I\\\I&0\end{pmatrix},\quad
 Z=\begin{pmatrix}1&0\\0&-1\end{pmatrix}.
\]
A \emph{three-qubit Pauli string} is a tensor product $P_A\otimes P_B\otimes P_C$ with each factor in $\{I,X,Y,Z\}$.  We abbreviate, for example, $X\otimes I\otimes Z$ as $XIZ$.

\begin{lemma}\label{lem:paulibasis}
The $64$ three-qubit Pauli strings form an orthogonal real basis of the Hermitian $8\times8$ matrices under the inner product
\[
 \langle A,B\rangle=\frac18\tr(AB).
\]
The $63$ nonidentity strings form an orthonormal basis of the traceless Hermitian matrices, which is used to identify the tangent space of $\SU(8)$ with $\R^{63}$.
\end{lemma}

\subsection{Smooth maps, differentials, and local universality}

A smooth matrix-valued map can be differentiated entry by entry.  At a point $x$, its differential $DF_x$ is the linear map that gives the first-order change in the output caused by a small change in the input.  After choosing local coordinates, $DF_x$ is the Jacobian matrix.  Its rank is independent of the coordinates chosen.

For a tangent vector $\dot U$ at $U\in\SU(8)$, define the right-trivialized Hermitian representative
\[
 \tau_U(\dot U)=\I\dot U U^\dagger.
\]
Because $\dot U U^\dagger$ is skew-Hermitian and traceless, $\tau_U(\dot U)$ is Hermitian and traceless.  Expanding it in the Pauli basis of \Cref{lem:paulibasis} gives a vector in $\R^{63}$.

\begin{definition}\label{def:accessible}
For a fixed circuit architecture with smooth circuit map $\Phi$, its \emph{accessible dimension} is
\[
 d=\max_x\rank D\Phi_x.
\]
The architecture is \emph{locally universal} if its reachable set contains a nonempty open subset of $\SU(8)$.
\end{definition}

By the inverse function theorem, full differential rank $63$ implies local universality.

\section{Circuit model and support words}\label{sec:circuitmodel}

\subsection{Embedded two-qubit gates}

Label the qubits $A,B,C$.  The possible two-qubit supports are
\[
 E=\{AB,AC,BC\}.
\]
For $V\in\SU(4)$, write $V^{AB}=V\otimes I_2$ and $V^{BC}=I_2\otimes V$.  The embedding $V^{AC}$ is obtained by swapping $B$ and $C$, applying $V$ to the first two tensor factors, and swapping back.  Each support gives a subgroup
\[
 H_e=\{V^e:V\in\SU(4)\}\subseteq\SU(8),\qquad e\in E.
\]
Its tangent space at the identity, represented by traceless Hermitian matrices, is denoted $\Lie_e$ and has dimension $15$.

Free one-qubit gates form the local subgroup
\[
 \Local=\{A\otimes B\otimes C:A,B,C\in\SU(2)\}\subseteq\SU(8).
\]
Its Lie algebra has dimension nine.

\begin{definition}[Support word and circuit map]
A support word is a finite sequence
\[
 w=e_1e_2\cdots e_m,\qquad e_j\in E.
\]
The index $j=1$ denotes the earliest gate, so the matrix product is written in reverse chronological order.  With free one-qubit gates, the normal-form circuit map is
\[
 \Phi_w:\Local\times\SU(4)^m\longrightarrow\SU(8),\qquad
 \Phi_w(L,U_1,\ldots,U_m)=L U_m^{e_m}\cdots U_1^{e_1}.
\]
Its image is denoted $\Circuit_w$ and its accessible dimension is $d(w)$.
\end{definition}

\begin{fact}\label{prop:normalform}
A circuit with arbitrary one-qubit gates before, after, and between the two-qubit slots of $w$ has the same reachable set $\Circuit_w$ as the normal-form map above.
\end{fact}
\begin{proof}
Consider one two-qubit gate $V^e$ and a local gate $K=A\otimes B\otimes C$ immediately to its right. Let $s$ be the spectator qubit not contained in
$e$, and write $K=K_eK_s$, where $K_e$ acts on the two qubits in $e$ and
$K_s$ on $s$. Since $K_s$ commutes with $V^e$,
$
V^eK=K_s(V^eK_e),
$
with $V^eK_e\in H_e$. Thus the supported part of each local layer can be
absorbed into the adjacent two-qubit gate, while its spectator part is
moved to the left. Iterating through the circuit leaves a single
leftmost local layer.
\end{proof}

\subsection{Run reduction}

\begin{definition}
Delete every support letter that is equal to the immediately preceding surviving letter.  The resulting word is the \emph{run reduction} $\bar w$, and
\[
 r(w)=|\bar w|
\]
is the reduced length.
\end{definition}

For example,
\[
 AB,AB,BC,BC,AB\longmapsto AB,BC,AB.
\]

\begin{fact}\label{prop:runreduction}
For every support word $w$,
\[
 \Circuit_w=\Circuit_{\bar w},\qquad d(w)=d(\bar w).
\]
\end{fact}

\subsection{Qubit relabeling and reversal}

\begin{proposition}\label{prop:symmetry}
Accessible dimension and local universality are unchanged by either of the following operations:
\begin{enumerate}[label=(\roman*),leftmargin=2em]
\item applying any permutation of the qubits to every support letter;
\item reversing the order of the support word.
\end{enumerate}
\end{proposition}
\begin{proof}
For a qubit permutation $\pi$, let $P_\pi$ be the corresponding permutation unitary on $\Hilb$.  Conjugation by $P_\pi$ maps $H_e$ diffeomorphically onto $H_{\pi(e)}$ and maps $\Local$ onto itself.  Therefore it maps $\Circuit_w$ diffeomorphically onto $\Circuit_{\pi(w)}$, preserving differential ranks and open sets.

For reversal, take inverses.  If
\[
 U=L U_m^{e_m}\cdots U_1^{e_1},
\]
then
\[
 U^{-1}=(U_1^{-1})^{e_1}\cdots(U_m^{-1})^{e_m}L^{-1}.
\]
Using \Cref{prop:normalform} to collect the final local gate shows that inversion maps $\Circuit_w$ diffeomorphically onto $\Circuit_{w^{\mathrm{rev}}}$.  Inversion is a smooth involution on $\SU(8)$, so it preserves accessible dimension and interior.
\end{proof}

\section{The upper bound}

\begin{lemma}\label{thm:upperbound}
For every support word $w$,
\[
d(w)
\le \min\{63,9r(w)+9\}.
\]
\end{lemma}

\Cref{thm:upperbound} implies that no architecture with five or fewer reduced two-qubit slots can be locally universal, because $54<63$.

\section{The differential in Pauli coordinates}

\subsection{Product-map derivative}

For a reduced support word $w=e_1\cdots e_m$, first ignore the free final local layer and consider
\[
 F_w:\SU(4)^m\to\SU(8),\qquad
 F_w(U_1,\ldots,U_m)=U_m^{e_m}\cdots U_1^{e_1}.
\]

\begin{lemma}[Right-trivialized differential]\label{lem:differential}
At a point $(U_1,\ldots,U_m)$, the right-trivialized differential image of $F_w$ is
\[
 \operatorname{im}DF_w
 =\sum_{j=1}^m
 \Ad_{S_j}(\Lie_{e_j}),
 \qquad
 S_j=U_m^{e_m}\cdots U_{j+1}^{e_{j+1}},
\]
where $S_m=I$ and $\Ad_S(H)=SHS^\dagger$.
\end{lemma}
\begin{proof}
Choose a traceless Hermitian generator $H\in\Lie_{e_j}$ and vary only the $j$-th gate by left multiplication:
\[
 U_j(t)^{e_j}=\exp(-\I tH)U_j^{e_j}.
\]
Write $F=F_w(U_1,\ldots,U_m)$ and $S_j=U_m^{e_m}\cdots U_{j+1}^{e_{j+1}}$.  Differentiating at $t=0$ gives
\[
 \dot F=-\I S_jH S_j^\dagger F.
\]
Right-trivialization yields
\[
 \tau_F(\dot F)=\I\dot F F^\dagger=S_jHS_j^\dagger=\Ad_{S_j}(H).
\]
Allowing $H$ to range over the $15$-dimensional space $\Lie_{e_j}$ gives the contribution of slot $j$.  The differential is linear in simultaneous variations, so its image is the sum of the slot subspaces.
\end{proof}

By \Cref{lem:paulibasis}, each conjugated generator has $63$ real Pauli coordinates.  Taking the $15$ generators on each support as columns produces a $63\times15m$ matrix $J_w$, whose rank equals $\rank DF_w$.

\subsection{Rational Pauli rotations}

For a nonidentity two-qubit Pauli string $P$, define
\[
 R_P(\theta)=\exp(-\I\theta P/2).
\]
Since $P^2=I$, the exponential is
\[
 R_P(\theta)=\cos(\theta/2)I-\I\sin(\theta/2)P.
\]

\begin{lemma}[Conjugation by a Pauli rotation]\label{lem:paulirotation}
Let $P$ and $Q$ be Pauli strings.  If they commute, then
\[
 R_P(\theta)QR_P(\theta)^\dagger=Q.
\]
If they anticommute, then
\[
 R_P(\theta)QR_P(\theta)^\dagger
 =\cos\theta\,Q-\I\sin\theta\,PQ.
\]
If $t=\tan(\theta/4)\in\Q$, then all coefficients in this formula are rational.
\end{lemma}
\begin{proof}
Write $c=\cos(\theta/2)$ and $s=\sin(\theta/2)$.  If $P$ and $Q$ commute, then $Q$ commutes with every polynomial in $P$, hence with $R_P(\theta)$.

If $PQ=-QP$, expand directly:
\begin{align*}
(cI-\I sP)Q(cI+\I sP)
&=c^2Q+\I csQP-\I csPQ+s^2PQP\\
&=(c^2-s^2)Q-2\I csPQ\\
&=\cos\theta\,Q-\I\sin\theta\,PQ.
\end{align*}
Finally, if $t=\tan(\theta/4)$, then
\[
 c=\frac{1-t^2}{1+t^2},\qquad
 s=\frac{2t}{1+t^2},
\]
so $c,s,\cos\theta=c^2-s^2$, and $\sin\theta=2cs$ are rational.
\end{proof}

Every gate used in the architecture certificates is a short product of such rotations with rational $t$.  Therefore every Pauli coordinate of every conjugated tangent generator is rational.

\subsection{Finite-field proof certificates}

\begin{lemma}\label{lem:finitefield}
Let $M$ be a matrix with rational entries.  Let $p$ be a prime that divides none of the denominators appearing in a selected square minor $N$.  Reduce every entry of $N$ modulo $p$ by replacing $a/b$ with $a b^{-1}\pmod p$.  If
\[
 \det(N)\not\equiv0\pmod p,
\]
then $\det(N)\ne0$ in $\Q$, and therefore $\rank M$ over $\R$ and $\C$ is at least the size of $N$.
\end{lemma}
\begin{proof}
Let $D$ be a common positive denominator of the entries of $N$, chosen so that $p\nmid D$.  Then $A=DN$ is an integer matrix and
$
 \det A=D^r\det N
$
for an $r\times r$ minor.  Reduction modulo $p$ gives
$
 \det(A\bmod p)\equiv D^r\det(N\bmod p)\pmod p.
$
Both $D^r$ and $\det(N\bmod p)$ are nonzero modulo $p$, so $\det A$ is not divisible by $p$ and in particular is a nonzero integer.  Hence $\det N=\det A/D^r$ is a nonzero rational number.
\end{proof}

\section{Exhaustive architecture classes and certificates}

\subsection{Canonical support words}

For compact notation set
\[
 0=AB,\qquad 1=AC,\qquad 2=BC.
\]
A word is reduced when adjacent digits are different.

\begin{observation}\label{prop:classes}
Up to qubit relabeling and reversal, the reduced support words of lengths one through six have the following canonical representatives:
\begin{center}
\begin{tabular}{c l}
\toprule
\textnormal{length} & \textnormal{representatives}\\
\midrule
$1$ & \textnormal{\word{0}}\\
$2$ & \textnormal{\word{01}}\\
$3$ & \textnormal{\word{010}, \word{012}}\\
$4$ & \textnormal{\word{0101}, \word{0102}, \word{0120}}\\
$5$ & \textnormal{\word{01010}, \word{01012}, \word{01020}, \word{01021}, \word{01201}, \word{01210}}\\
$6$ & \textnormal{\word{010101}, \word{010102}, \word{010120}, \word{010121}, \word{010201},}\\
  & \textnormal{\word{010210}, \word{010212}, \word{012012}, \word{012021}, \word{012120}.}\\
\bottomrule
\end{tabular}
\end{center}
There are $22$ classes of lengths two through six.
\end{observation}

\subsection{Certificate format}

For each canonical word of length $m\ge2$, we specify:
\begin{enumerate}[leftmargin=2.2em]
\item a list of $m$ two-qubit gates, each given as a product of Pauli rotations;
\item rational tangent-half-angle data $t=\tan(\theta/4)$ for each rotation;
\item a set of row and column indices selecting a square Jacobian minor;
\item the determinant of that minor modulo $p=1{,}000{,}003$.
\end{enumerate}
The Pauli strings are ordered lexicographically with $I<X<Y<Z$.  The global rows are the $63$ strings in $\{I,X,Y,Z\}^3\setminus\{III\}$, and the $15$ columns for each slot are the strings in $\{I,X,Y,Z\}^2\setminus\{II\}$.  Gate slots are ordered from earliest to latest.  If a gate recipe lists $(P_1,t_1),\ldots,(P_s,t_s)$, the gate is
\[
 U=R_{P_s}(\theta_s)\cdots R_{P_1}(\theta_1),
 \qquad t_a=\tan(\theta_a/4).
\]

\begin{table}[H]
\centering
\caption{All architecture certificates.  Every residue is the determinant of the selected maximal minor modulo the prime $1{,}000{,}003$.}
\label{tab:certsummary}
\begin{tabular}{@{}cccc@{}}
\toprule
canonical word & length & certified rank & residue\\
\midrule
01 & 2 & 27 & 692520 \\
010 & 3 & 36 & 767294 \\
012 & 3 & 36 & 201033 \\
0101 & 4 & 45 & 241080 \\
0102 & 4 & 45 & 153904 \\
0120 & 4 & 45 & 986707 \\
01010 & 5 & 54 & 758597 \\
01012 & 5 & 54 & 833372 \\
01020 & 5 & 54 & 30187 \\
01021 & 5 & 54 & 52097 \\
01201 & 5 & 54 & 4519 \\
01210 & 5 & 54 & 30415 \\
010101 & 6 & 63 & 331488 \\
010102 & 6 & 63 & 2796 \\
010120 & 6 & 63 & 190332 \\
010121 & 6 & 63 & 349679 \\
010201 & 6 & 63 & 518497 \\
010210 & 6 & 63 & 728400 \\
010212 & 6 & 63 & 899673 \\
012012 & 6 & 63 & 794691 \\
012021 & 6 & 63 & 91876 \\
012120 & 6 & 63 & 822563 \\
\bottomrule
\end{tabular}
\end{table}

\begin{proposition}\label{prop:certifiedranks}
For every reduced support word $w$ of length $m\in\{2,3,4,5,6\}$,
\[
 \max\rank DF_w=\min\{63,9m+9\}.
\]
\end{proposition}
\begin{proof}
By \Cref{prop:symmetry} and \Cref{prop:classes}, it is enough to prove one representative of each row of \Cref{tab:certsummary}.  The complete rational recipes, pivot rows, pivot columns, and residues are reproduced in the machine-readable certificate in \Cref{app:manifest}.

The verifier performs the following steps for each representative.
It constructs every tangent column using the formula in \Cref{lem:differential} and the rational conjugation rule in \Cref{lem:paulirotation}.  It checks that every denominator is invertible modulo $1{,}000{,}003$. It extracts the selected square minor and recomputes its determinant by modular Gaussian elimination.  The resulting determinant is the residue displayed in \Cref{tab:certsummary}.  By \Cref{lem:finitefield}, the corresponding rational minor is nonzero, so the real Jacobian rank is at least the displayed value.
The upper bound in \Cref{thm:upperbound} is equal to the displayed value for each length.  Therefore the lower and upper bounds agree.
\end{proof}

Combining the preceding results yields the following complete dimension law.
\begin{theorem}\label{thm:main}
For every three-qubit support word $w$,
\[
d(w)=\min\{63,9r(w)+9\}.
\]
\end{theorem}

\subsection{A particularly simple fixed-path witness}

Our exhaustive certificate proves that every alternating two-edge word works.  The following separate witness is conceptually useful because the same two-qubit matrix is repeated in all six slots.

\begin{proposition}[Repeated-gate full-rank witness on a fixed line]\label{prop:fixedwitness}
For the support word
\[
 w=AB,BC,AB,BC,AB,BC,
\]
there is one two-qubit gate $\widetilde G\in\SU(4)$ such that setting all six slots equal to $\widetilde G$ gives $\rank DF_w=63$.
\end{proposition}
\begin{proof}
Let
\[
 H=\begin{pmatrix}
0&0&-1+\I&0\\
0&1&0&0\\
-1-\I&0&1&\I\\
0&0&-\I&1
\end{pmatrix}.
\]
The matrix is Hermitian.  Since every eigenvalue of $H$ is real, $I+\I H$ is invertible.  Define the Cayley transform
\[
 G=(I-\I H)(I+\I H)^{-1}
 =\frac1{25}\begin{pmatrix}
-3+4\I&0&16+12\I&14-2\I\\
0&-25\I&0&0\\
-12+16\I&0&-11-2\I&6-8\I\\
-2-14\I&0&-6+8\I&1-18\I
\end{pmatrix}.
\]
Because $I-\I H$ and $I+\I H$ are commuting polynomials in $H$,
\[
 G^\dagger=(I-\I H)^{-1}(I+\I H),
\]
and hence $G^\dagger G=I$.  Direct determinant expansion of $G$ gives
\[
 \det G=\frac{-24+7\I}{25},
\]
which has modulus one.  Choose any scalar $\lambda$ satisfying $\lambda^4=(\det G)^{-1}$ and set $\widetilde G=\lambda G$.  Then $\widetilde G\in\SU(4)$.  Scalar multiplication does not change conjugation, so $\Ad_{\widetilde G}=\Ad_G$ and the Jacobian rank is unchanged.

The Jacobian is computed in the $63$-element Pauli basis using \Cref{lem:differential}.  Select columns
\begin{align*}
&0,1,2,3,4,5,6,7,8,9,10,11,12,13,14,18,19,20,21,22,23,24,25,26,27,28,29,\\
&34,35,36,38,39,40,42,43,44,49,50,51,53,54,55,57,58,59,64,65,66,68,69,70,\\
&72,73,74,79,80,81,83,84,85,87,88,89.
\end{align*}
The resulting $63\times63$ minor has determinant residues
\[
\begin{array}{c|ccccc}
 p&101&109&113&137&149\\\hline
 \\[-10pt]
 \sqrt{-1}\pmod p&10&33&15&37&44\\
 \det&50&97&16&18&141
\end{array}
\]
respectively.  Every residue is nonzero.  The entries are Gaussian rationals with denominator a power of $25$, and none of these primes divides $25$.  The Gaussian-integer version of \Cref{lem:finitefield} is obtained by reducing $\I$ to the square root of $-1$.  Therefore the minor is nonzero over $\Q(\I)$ and the real Jacobian rank is $63$.
\end{proof}

\section{Discussion}

\subsection{Interpretation for circuit compilation}

The full-rank result has a direct interpretation for circuit compilation.  At an explicit six-gate witness, the differential spans all $63$ tangent directions of $\SU(8)$.  Hence targets sufficiently close to the corresponding circuit output can, in principle, be compiled by solving a locally well-posed system of $63$ equations. The theorem rules out a structural rank deficiency of the six-block template. It does not, however, imply that an arbitrary numerical optimizer will recover suitable parameters from an arbitrary initialization.

The fixed-path result is particularly relevant for hardware-constrained compilation.  Local universality does not require changing the pairwise placement of the gates or introducing SWAP gates merely to obtain sufficient local expressivity.  In particular, the nearest-neighbor architecture already attains full rank with six blocks.  Thus the parameter-count threshold is attainable even on a fixed three-qubit line. In this sense, six-gate local universality is architecture-robust at the level of support words.

The conclusion is nevertheless local rather than global. What the present result establishes is that any obstruction to global synthesis with six $\SU(4)$ gates cannot arise from insufficient parameter count, nearest-neighbor connectivity, or generic local differential rank. Such an obstruction would have to be genuinely global.

\subsection{Open problems}

The distinction between local and global universality leaves several concrete directions for further work.

\begin{enumerate}[leftmargin=2.2em]

\item \textbf{Global six-gate surjectivity.}  Determine whether some, every, or no reduced six-slot support word has image equal to $\SU(8)$.

\item \textbf{Constructive inversion.}  Full rank guarantees local solvability but does not itself give a practical compiler.  Developing a stable symbolic or numerical procedure that maps a target in a certified open neighborhood to six $\SU(4)$ parameters, ideally with a convergence guarantee, would turn the local universality theorem into a constructive synthesis result.

\item \textbf{More qubits.}  It is natural to ask for analogous accessible-dimension laws for fixed architectures on four or more qubits. Exhaustive support enumeration and Jacobian certificates remain conceivable, but both the number of architectures and the size of the required certificates grow substantially.

\end{enumerate}

\section{Conclusion}

The six-gate parameter threshold for a general three-qubit unitary has been visible for decades, but parameter counting alone could not show that the threshold is attainable.  The certificates in this work close that local gap.  Six arbitrary two-qubit gates have a full-rank differential on a fixed three-qubit line, and this is not an accident of one arrangement: every reduced six-slot support word is locally universal.

The result gives a minimal, architecture-robust local template for three-qubit compilation.  It also sharpens the remaining global problem.  If six gates ever fail to synthesize some three-qubit unitaries, the reason is not missing parameters, nearest-neighbor connectivity, or local Jacobian rank.  The obstruction must be global.

\appendix
\crefalias{section}{appendix}

\section{Human-readable certificate}\label{app:recipes}

The notation $PQ(t)$ means the two-qubit Pauli rotation $R_{P\otimes Q}(\theta)$ with $t=\tan(\theta/4)$.  Entries on each gate row are listed in the order stored by the certificate. The rightmost listed rotation is leftmost in the resulting matrix product.

\footnotesize
\begin{longtable}{@{}L{18mm}C{16mm}L{100mm}@{}}
\caption{Pauli-rotation recipes for all canonical architecture certificates.}\label{tab:recipes}\\
\toprule
Word & Rank & Gates\\
\midrule
\endfirsthead
\toprule
Word & Rank & Gates\\
\midrule
\endhead
\texttt{01} & 27 & \raggedright $U_{1}$: XY(1/3), YY(1/5); $U_{2}$: ZZ(1/4), XI(-1/2)\arraybackslash \\
\texttt{010} & 36 & \raggedright $U_{1}$: XI(-2/3), ZZ(1/4); $U_{2}$: ZX(-1/4), XZ(-1/2); $U_{3}$: IX(-1/3), IX(2/5)\arraybackslash \\
\texttt{012} & 36 & \raggedright $U_{1}$: IY(1/3), XX(4/5); $U_{2}$: XY(-3/4), YI(4/5); $U_{3}$: ZI(3/5), ZI(1/2)\arraybackslash \\
\texttt{0101} & 45 & \raggedright $U_{1}$: YY(2/3), XZ(-2/3); $U_{2}$: ZY(1/4), YZ(2/3); $U_{3}$: XY(1/2), YZ(3/5); $U_{4}$: ZI(1/5), YI(3/5)\arraybackslash \\
\texttt{0102} & 45 & \raggedright $U_{1}$: ZY(3/4), ZI(3/5); $U_{2}$: YZ(-1/4), XY(2/5); $U_{3}$: ZX(1/4), YX(-3/4); $U_{4}$: ZI(2/3), XI(1/3)\arraybackslash \\
\texttt{0120} & 45 & \raggedright $U_{1}$: YX(3/4), XX(-1/4); $U_{2}$: XZ(-2/3), YY(3/5); $U_{3}$: ZZ(1/3), XY(-3/4); $U_{4}$: ZZ(-1/4), ZI(3/5)\arraybackslash \\
\texttt{01010} & 54 & \raggedright $U_{1}$: XI(3/4), ZZ(-1/3), XZ(-2/3); $U_{2}$: ZX(3/4), XZ(2/3), YY(-1/4); $U_{3}$: IZ(-2/3), YX(-2/3), ZY(-1/4); $U_{4}$: XX(-1/3), YI(1/2), XY(-1/2); $U_{5}$: ZZ(4/5), XZ(-3/4), YY(1/2)\arraybackslash \\
\texttt{01012} & 54 & \raggedright $U_{1}$: IZ(-1/2), ZY(-3/4); $U_{2}$: XZ(1/3), YX(-1/4); $U_{3}$: YZ(3/4), ZY(-1/3); $U_{4}$: XX(1/3), ZZ(2/5); $U_{5}$: XI(-1/4), IX(3/5)\arraybackslash \\
\bottomrule
\end{longtable}

\section{Machine-readable certificate}\label{app:manifest}

%

\bibliographystyle{unsrt}
\bibliography{main}

\end{document}